\documentclass[a4paper,10pt]{article}
\usepackage[utf8]{inputenc}

\usepackage{amsmath}
\usepackage{amsthm}
\usepackage{amssymb}

\usepackage{graphicx}
\usepackage{wrapfig}
\usepackage[table]{xcolor}
\usepackage{tikz}

\newtheorem{theorem}{Theorem}[section]
\newtheorem*{conjecture}{Conjecture}

\newcommand{\fp}{\operatorname{fp}}

\title{Slowly synchronizing DFAs of 7 states and maximal slowly synchronizing DFAs}
\author{Michiel de Bondt}

\begin{document}

\maketitle

\begin{abstract}
We compute all synchronizing DFAs with $7$ states and synchronization length $\ge 29$.

Furthermore, we compute alphabet size ranges for maximal, minimal and semi-minimal 
synchronizing DFAs with up to $7$ states.
\end{abstract}

\section{Introduction}

A \emph{Deterministic Finite Automaton (DFA)} consists of a finite set of
so-called \emph{states}, and a finite alphabet of so-called 
\emph{transition symbols}. The transition symbols are maps from the state set 
to itself. A DFA also has a begin state and a set of final states, but those
are irrelevant for this paper.

Let $Q$ and $\Sigma$ be the state set and the alphabet of a DFA. Then
the maps of the transition symbols combine to a \emph{transition function}
from $Q \times \Sigma$ to $\Sigma$. We denote this function by $\cdot$.
$\cdot$ is left-associative, and we will omit it mostly. We additionally 
define $\cdot: 2^Q \times \Sigma \rightarrow 2^Q$, namely
by $S x = \bigcup_{s \in S} \{sx\}$.

With $*$ being the Kleene star, $\Sigma^{*}$ is the sets of al words over
$\Sigma$, i.e.\@ all sequences of zero or more symbols of $\Sigma$. 
Each such word can be seen as either the empty word, or a symbol followed
by another word. With respect to this structural definition, we 
define $qw$ inductively as follows for states $q \in Q$, subsets 
$S \subseteq Q$, and words $w \in \Sigma^{*}$:
\begin{align*}
q \lambda &= q & q (x w) &= (qx)w & S \lambda &= S & S (x w) &= (Sx)w
\end{align*}
Here, $\lambda$ is the empty word, $x$ is the first letter of the word 
$xw$ and $w$ is the rest of $xw$. 

We say that a DFA with state set $Q$ and alphabet $\Sigma$ is
\emph{synchronizing} (in $l$ steps), if there exists a $w \in \Sigma^{*}$ 
(of length $l$), such that $Qw$ has size $1$.
\begin{wrapfigure}{r}{4cm}
\vspace*{-20pt}
\begin{tikzpicture}
\tikzstyle{nodestyle}=[draw,fill=white,circle,inner sep=0pt,minimum width=0.5cm]
\foreach[count=\n] \a in {240,180,...,-60} {
  \node[nodestyle] at (\a:1.1) (\n) {};
}
\draw (1) edge[->] node[auto,inner sep=1pt] {$a$} (2) 
  (2) edge[->] node[auto,inner sep=1pt] {$a$} (3) 
  (3) edge[->] node[auto,inner sep=1pt] {$a,b$} (4)
  (4) edge[->] node[auto,inner sep=1pt] {$a$} (5) 
  (5) edge[->] node[auto,inner sep=1pt] {$a$} (6);
\foreach \n in {1,2,4,5,6} {
  \draw (\n) edge[in=-60*\n-100,out=-60*\n-20,looseness=7,->] (\n);
  \draw (-60*\n-60:1.6) node {$b$};
}
\draw (-0.1,-0.9526) node {$\cdot$} (0.0,-0.9526) node {$\cdot$} (0.1,-0.9526) node {$\cdot$};
\end{tikzpicture}
\vspace*{-40pt}
\end{wrapfigure}
If a DFA is synchronizing
in $l$ steps, but not in fewer than $l$ steps, then we call $l$ the 
\emph{synchronization length}. A conjecture by \v{C}ern\'y in 1964 \cite{cerny}
is that for a DFA with $n$ states, the largest possible synchronization length 
is $(n-1)^2$. \v{C}ern\'y constructed a series of DFAs which reach this 
synchronization length, which is depicted on the right.
The unique shortest synchronizing word of this DFA is $b(a^{n-1}b)^{n-2}$.

In section \ref{dfa7}, we will discuss our search for DFAs with 7 states and large 
synchronization lengths. This search extends results in \cite{KKS16} and \cite{BDZ18}.
To obtain a more efficient search algorithm, we improved the pruning of \cite{BDZ18}.

In section \ref{max}, we will define several types of synchronizing DFAs, and we
discuss the search for DFAs with up to 7 states of these types. Some of these types
were already discussed in \cite{BDZ16}, in which the search has already been done
for DFAs with synchronization length $(n-1)^2$.

\section{Slowly synchronizing DFAs with 7 states} \label{dfa7}

In \cite{BDZ18}, we computed all basic DFAs with 7 states with synchronization
length at least $31$. This yielded only $22$ DFAs up to reordering states. 
Using better pruning, but pruning which only works for DFAs and not for 
PFAs in general, we extended this computation to synchronization 
length at least $29$, yielding no less than $1850647$ DFAs up to reordering 
states. The results are given below. 
Similar computations for smaller state sets can be found in \cite{BDZ16}.
\begin{center}
  \renewcommand{\arraystretch}{1.4}
  \begin{tabular}{|r|rrrrrrrr|}
    \hline
    alph.\@ & sync.\@ & sync.\@ & sync.\@ & sync.\@ & sync.\@ & sync.\@ & sync.\@ & sync.\@ \\[-5pt]
    size & 36 & 35 & 34 & 33 & 32 & 31 & 30 & 29 \\
    \hline
    1 & & & & & & & & \\[-5pt]
    2 & 1 & & & & 3 & 3 & 13 & 39 \\[-5pt]
    3 & & & & & 3 & 8 & 44 & 373 \\[-5pt]
    4 & & & & & & 4 & 90 & 1902 \\[-5pt]
    5 & & & & & & & 148 & 7416 \\[-5pt]
    6 & & & & & & & 194 & 23486 \\[-5pt]
    7 & & & & & & & 183 & 60544 \\[-5pt]
    8 & & & & & & & 113 & 126448 \\[-5pt]
    9 & & & & & & & 44 & 213970 \\[-5pt]
    10 & & & & & & & 10 & 294678 \\[-5pt]
    11 & & & & & & & 1 & 331780 \\[-5pt]
    12 & & & & & & & & 306068 \\[-5pt]
    13 & & & & & & & & 231142 \\[-5pt]
    14 & & & & & & & & 142256 \\[-5pt]
    15 & & & & & & & & 70713 \\[-5pt]
    16 & & & & & & & & 27980 \\[-5pt]
    17 & & & & & & & & 8620 \\[-5pt]
    18 & & & & & & & & 2000 \\[-5pt]
    19 & & & & & & & & 332 \\[-5pt]
    20 & & & & & & & & 36 \\[-5pt]
    21 & & & & & & & & 2 \\
    \hline
    total & 1 & 0 & 0 & 0 & 6 & 15 & 840 & 1849785 \\
    \hline
  \end{tabular}
\end{center}
The computation took $8.5$ CPU-years on a heterogeneous cluster, and the estimated
single-thread time was about $5$ years. The computation was performed by borrowing
CPU-cyles from the science department of our university, especially the theoretical
chemistry group.

There exists a basic DFA with $7$ states, $39$ ($53$) symbols, and synchronization length
$28$ ($27$). This shows that enumeration of basic DFA with $7$ states and synchronization 
length $28$ ($27$) is not feasible. In the next section, we suggest computations which can
be performed in practice instead.
 
As mentioned above, the algorithm differs from that in \cite{BDZ18} 
in that the pruning has been improved. The pruning is done by finding an 
upper bound of the synchronization length of all synchronizing extensions 
$\mathcal{B}$ of a DFA $\mathcal{A}$. Here, $\mathcal{B}$ is an extension 
of $\mathcal{A}$ if $\mathcal{A}$ and $\mathcal{B}$ have the same state
sets, and for every symbol $a$ of $\mathcal{A}$, there exists a symbol $a'$ of 
$\mathcal{B}$ which corresponds to $a$ as a (partial) mapping of states.

The pruning in \cite{BDZ18} comes in three variants, with three upper bound
$L$, $L'$, and $L''$. The first variant is the easiest.
\begin{enumerate}
	\item[(1)] Determine the size $|S|$ of a smallest reachable set $S$. 
	Let $m$ be the minimal distance from $Q$ to a set of size $|S|$.
	\item[(2)] For each $k\leq |S|$, partition the collection of irreducible sets of size $k$ into strongly
	connected components. Let $m_k$ be the number of components plus the sum of their diameters.
	\item[(3)] For each reducible set $R$ of size $k\leq |S|$, find the length $l_R$ of its shortest
	reduction word. Let $l_k$ be the maximum of these lengths.
	\item[(4)] Now note that a synchronizing extension of $\mathcal{A}$ will have a synchronizing
	word of length at most
	\[ L \; = \; \sum_{k=2}^{|S|}(m_k+l_k) + m. \]
\end{enumerate}
The second variant improves the first variant as follows. Let $M$ be the maximum 
distance from $Q$ to a set of size $|S|$. Partition the irreducible sets of size $|S|$ 
which can be reached from $Q$ into strongly connected components, and let $c$ 
be the number of components plus the sum of their diameters. Then
a synchronizing extension of $\mathcal{A}$ will have a synchronizing
word of length at most
\[ L' \; = \; \sum_{k=2}^{|S|}(m_k+l_k) - c + 1 + M. \]
The third variant is the hardest variant. We take the upper bound $L''$ 
equal to $L''_Q$, and we define inductively an upper bound $L''_R$
for the length of the the shortest synchronizing word for 
a reducible subset $R$, and an upper bound $L''_k$ for the maximum
length of the shortest synchronizing word for any subset of size $k$.
Define $S_R$, $m_R$, $M_R$ and $c_R$ as $S$, $m$, $M$ and $c$ respectively, 
but with $Q$ replaced by $R$.
\begin{align*}
L''_R \; &= \; m_R && \mbox{if $|S_R| = 1$,} \\
L''_R \; &= \; \min\{ L''_{|S_R|} - c_R + 1 + M_R, L''_{|R|-1} + l_R \} 
                   && \mbox{if $|S_R| > 1$,} \\
L''_1 \; &= \; 0, \\
L''_k \; &= \; m_k + \max\{ L''_{k-1}, L''_R \mid R 
\mbox{ is reducible and } |R| = k \} && \mbox{if $k > 1$.}
\end{align*}

We improve the three upper bounds as follows.
\begin{itemize}

\item In $L$, we improve $m_k$ for each $k \le |S|$;

\item In $L'$, we improve $M$, $m_{|S|} - c$ and $m_k$ for each $k < |S|$;

\item In $L''$, we improve $M_R$ and $m_{|S_R|} - c_R$ for each $R \subseteq Q$
and $m_k$ for each $k < |S|$.

\end{itemize}
Since $M = M_Q$, $S = S_Q$ and $c = c_Q$, it suffices to improve $m_k$ for each 
$k \le |S|$, and $M_R$ and $m_{|S_R|} - c_R$ for each $R \subseteq Q$. We must
preserve the following.
\begin{enumerate}

\item[($\alpha$)] Let $k \le |S|$. For every synchronizing extension $\mathcal{B}$ 
of $\mathcal{A}$, the shortest path from any subset of size
$k$ to a subset of size $\le k$ which is either reducible in $\mathcal{A}$
or of size $< k$, has length at most $m_k$.

\item[($\beta$)] Let $R \subseteq Q$ be reducible in $\mathcal{A}$. Notice that 
$|S_R|$ is the size of the smallest set which is reachable from $R$ in 
$\mathcal{A}$. For every synchronizing extension $\mathcal{B}$ of 
$\mathcal{A}$, the shortest path from $R$ to a subset of size $\le |S_R|$ 
which is reducible in $\mathcal{A}$ or of size $< |S_R|$, has 
length at most $M_R + 1 + m_{|S_R|} - c_R$.

\end{enumerate}
The first improvement is obtained by realizing that for subsets of size $k$
in ($\alpha$) and of size $|S_R|$ in ($\beta$) which are not reducible in 
$\mathcal{A}$, the only thing that matters is that they contain a pair, 
from which there exists a short path in $\mathcal{B}$ to a subset of size 
$\le 2$ which is either reducible in $\mathcal{A}$ or of size $< 2$.

For $m_k$, the improvement is as follows.
Let $\sigma$ be a strongly connected component of irreducible subsets 
of size $k$ of the power automaton of $\mathcal{A}$. The purpose of 
$m_k$ is to estimate the number of subsets of $\sigma$ in a synchronization 
path of the power automaton of $\mathcal{B}$, which is done by the diameter
of $\sigma$, i.e.\@
$$
\max \Big\{ \max \Big\{ d(S_1,S_2) \,\Big|\, S_2 \in \sigma \Big\} 
\,\Big|\, S_1 \in \sigma \Big\}
$$
where $d(S_1,S_2)$ is the number of steps required to get from $S_1$ to $S_2$
in $\mathcal{A}$. This can be improved to
$$
\max \Big\{ \max \Big\{ \min \Big\{ d(S_1,T) \,\Big|\, T \supseteq P \Big\} \,\Big|\,
\mbox{ \small \begin{tabular}{@{}c@{}} $P$ is a pair contained in \\ 
in some subset $S_2$ of $\sigma$ \end{tabular} } \Big\} \,\Big|\, S_1 \in \sigma \Big\}
$$
The purpose of $c_R$ is to exclude some strongly connected components which 
are considered in $m_{|S_R|}$, which can be done in the same way as before.

For $M_R$, the improvement is as follows. Let $\tau$ be the collection of subset of 
size $|S_R|$ which are reachable from $R$ in $\mathcal{A}$. Then we can
improve
$$
\max \Big\{ d(R,T) \,\Big|\, T \in \tau \Big\} 
$$
to
$$
\max \Big\{ \min \Big\{ d(R,T) \,\Big|\, T \supseteq P \Big\} \,\Big|\,
\mbox{$P$ is a pair contained in some subset of $\tau$}\Big\} 
$$

For the second improvement, we use ideas of \cite{KS14} and \cite{KKS16}.
Let $k \ge 2$, and $S_1, S_2, \ldots, S_{\ell}$ be distinct $k$-subsets and 
$P_1, P_2, \ldots, P_{\ell}$ be distinct pairs of states. We say that 
$$
(S_1,P_1), (S_2,P_2), \ldots, (S_{\ell},P_{\ell})
$$
is a \emph{Frankl-Pin sequence}, if
\begin{enumerate}

\item[(i)] $P_i \subseteq S_i$ for all $i$;

\item[(ii)] $P_j \nsubseteq S_i$ for all $i$ and all $j < i$.

\end{enumerate}
Let $\sigma$ be a collection of $k$-subsets of states and let $\pi$ be
a collection of pairs of states. Denote by $\fp(\rho)$ ($\fp(\sigma,\pi)$) 
the length of the longest Frankl-Pin sequence
$(S_1,P_1), (S_2,P_2), \ldots, (S_{\ell},P_{\ell})$, with
$S_i \in \rho$ (and $P_i \in \pi$) for all $i$. 

\begin{theorem} \label{fpks}
Let $\rho_k$ be the collection of $k$ subsets of states which are
reducible in $\mathcal{A}$. 
\begin{enumerate}

\item[\upshape(i)] Let $T$ be a subset of states of size $k$. In $\mathcal{B}$,
it takes at most
$$
\binom{n-k+2}{2} - \fp(\rho_k,\rho_2)
$$
steps to get from $T$ to a subset of size $\le k$ which is either reducible 
in $\mathcal{A}$ or of size $< k$.

\item[\upshape(ii)] Let $\tau$ be a collection of subsets of size $k$. Then 
there exists a $T \in \tau$, such that in $\mathcal{B}$, it takes at most
$$
1 - \fp(\tau) + \binom{n-k+2}{2} - \fp(\rho_k,\rho_2)
$$
steps to get from $T$ to a subset of size $\le k$ which is either reducible 
in $\mathcal{A}$ or of size $< k$.

\end{enumerate}
\end{theorem}

\begin{proof}
The proof of (i) is essentially that of \cite[Theorem 2]{KS14} and 
\cite[Theorem 1]{KKS16}, and the proof of (ii) is similar.
\end{proof}

Notice that Theorem \ref{fpks} (i) is a special case of Theorem \ref{fpks} (ii), 
namely the case where $|\tau| = 1$. On account of Theorem \ref{fpks} (i), 
we can improve $m_k$ to 
$$
\min \Big\{m_k, \binom{n-k+2}{2} - \fp(\rho_k,\rho_2)\Big\}
$$
On account of Theorem \ref{fpks} (ii), we can improve $m_k - c_R$ with $k = |S_r|$ to
$$
\min \Big\{m_k - c_R, -\fp(\tau) + \binom{n-k+2}{2} - \fp(\rho_k,\rho_2)\Big\}
$$
There is however one problem, namely computing $\fp(\tau)$ and $\fp(\rho_k,\rho_2)$.
We do not compute $\fp(\tau)$ and $\fp(\rho_k,\rho_2)$, but take the lengths of
Frankl-Pin sequences which are not necessarily maximal. This makes the 
improvements of $m_k$ and $m_{S_R} - c_R$ worse, but they remain valid. 

We construct the Frankl-Pin sequences with length $\le \fp(\sigma,\pi)$ by a 
greedy approach. We take a pair $P$ of $\pi$ which is contained in the fewest 
subsets of $\sigma$. We make $\sigma'$ from $\sigma$ by removing all subsets which 
contain $P$. We compute a lower bound $f$ of $\fp(\sigma',\pi \setminus \{P\})$ 
recursively. If $\sigma' \ne \sigma$, then the Frankl-Pin sequence with length
$f$ can be extended at the front, and $1 + f$ is a lower bound of 
$\fp(\sigma,\pi)$. If $\sigma' = \sigma$, then $f$ is a lower bound of 
$\fp(\sigma,\pi)$.

A DFA is \emph{transitive} or \emph{strongly connected} if one can get from any 
state to any other state.
A synchronizing DFA is \emph{minimal} or \emph{irreducibly synchronizing}
if it becomes nonsynchronizing after removing any symbol. 
The authors of \cite{KS14} and \cite{KKS16} count the synchronizing automata 
differently, namely they count only transitive minimal synchronizing DFAs up to
reordering states. Below, we do this as well for $7$ states. 
\begin{center}
  \renewcommand{\arraystretch}{1.4}
  \begin{tabular}{|r|rrrrrrrr|}
    \hline
    alph.\@ & sync.\@ & sync.\@ & sync.\@ & sync.\@ & sync.\@ & sync.\@ & sync.\@ & sync.\@ \\[-5pt]
    size & 36 & 35 & 34 & 33 & 32 & 31 & 30 & 29 \\
    \hline
    1 & & & & & & & & \\[-5pt]
    2 & 1 & & & & 3 & 3 & 13 & 39 \\[-5pt]
    3 & & & & & & 2 & 29 & 257 \\[-5pt]
    4 & & & & & & & 8 & 145 \\[-5pt]
    5 & & & & & & & 4 & 55 \\[-5pt]
    6 & & & & & & & 1 & 4 \\
    \hline
    total & 1 & 0 & 0 & 0 & 3 & 5 & 55 & 500 \\
    \hline
  \end{tabular}
\end{center}
Actually, all slowly synchronizing minimal DFAs with $7$ states are counted 
above, because nontransitive synchronizing DFAs with $7$ states have 
synchronization length at most $26$.

\begin{theorem}
If the Cerny conjecture is true for less than $n \ge 2$ states, then the maximum length 
of the synchronizing word of a nontransitive synchronizing DFA with $n$ states
is 
$$
\max\{\tfrac12 n(n-1), (n-2)^2 + 1\}
$$
which is $(n-2)^2 + 2$ if $n = 3$ or $n = 4$, and $(n-2)^2 + 1$ otherwise.
\end{theorem}

\begin{proof}
Let $\mathcal{A}$ be a synchronizing DFA with $n$ states, and suppose that 
$\mathcal{A}$ has exacly $m$ states which can be reached from every other state.
Suppose that the Cerny conjecture holds for $m$ states. Then these $m$ states
can be synchronized in at most $(m-1)^2$ steps. It takes
$$
1 + 2 + \cdots + n-m = \tfrac12(n-m+1)(n-m)
$$
steps to reduce the set of all $n$ states to those $m$ states, so
the synchronization length is at most
$$
f(m) = (m-1)^2 + \tfrac12(n-m+1)(n-m)
$$
It is straightforward to show that $f(m)$ can indeed be obtained as 
a synchronization length. Since $f$ is a convex function, its maximum
is obtained at $m = 1$ or $m = n - 1$.
\end{proof}

In figure \ref{om}, we count transitive minimal synchronizing DFAs 
up to reordering states for less than $7$ states.

\begin{figure}[p] 
\rotatebox{90}{\begin{minipage}{\textheight}\small
\begin{center}
  \renewcommand{\arraystretch}{1.4}
  \begin{tabular}{|r|rrrr|r|}
    \hline
    alph.\@ & sync.\@ & sync.\@ & sync.\@ & sync.\@ & total \\[-5pt]
    size & 4 & 3 & 2 & 1 & \\
    \hline
    1 & & & & & 0 \\[-5pt]
    2 & 2 & 3 & 3 & & 8 \\[-5pt]
    3 & 2 & & & & 2 \\
    \hline
    total & 4 & 3 & 3 & 0 & 10 \\
    \hline
  \end{tabular}
  ~~
  \renewcommand{\arraystretch}{1.4}
  \begin{tabular}{|r|rrrrrrrrr|r|}
    \hline
    alph.\@ & sync.\@ & sync.\@ & sync.\@ & sync.\@ & sync.\@ & sync.\@ & sync.\@ & sync.\@ & sync.\@ & total \\[-5pt]
    size & 9 & 8 & 7 & 6 & 5 & 4 & 3 & 2 & 1 & \\
    \hline
    1 & & & & & & & & & & 0 \\[-5pt]
    2 & 2 & 5 & 11 & 20 & 49 & 52 & 57 & 18 & & 214 \\[-5pt]
    3 & 2 & 19 & 50 & 113 & 114 & 188 & 84 & & & 570 \\[-5pt]
    4 & & 2 & 5 & 5 & 4 & & & & & 16 \\
    \hline
    total & 4 & 26 & 66 & 138 & 167 & 240 & 141 & 18 & 0 & 800 \\
    \hline
  \end{tabular}
\end{center}
\begin{center}
  \renewcommand{\arraystretch}{1.4}
  \begin{tabular}{|r|rrrrrrrrrrrrrrrr|r|}
    \hline
    alph.\@ & sync.\@ & sync.\@ & sync.\@ & sync.\@ & sync.\@ & sync.\@ & sync.\@ & sync.\@ & sync.\@ & sync.\@ & sync.\@ & sync.\@ & sync.\@ & sync.\@ & sync.\@ & sync.\@ & total \\[-5pt]
    size & 16 & 15 & 14 & 13 & 12 & 11 & 10 & 9 & 8 & 7 & 6 & 5 & 4 & 3 & 2 & 1 & \\
    \hline
    1 & & & & & & & & & & & & & & & & & 0 \\[-5pt]
    2 & 1 & 4 & 11 & 23 & 43 & 46 & 139 & 224 & 380 & 622 & 986 & 1514 & 1547 & 893 & 99 & & 6532 \\[-5pt]
    3 & 1 & 8 & 31 & 89 & 448 & 841 & 1833 & 3892 & 7461 & 13471 & 23144 & 30931 & 27044 & 8344 & & & 117538 \\[-5pt]
    4 & & 1 & 4 & 42 & 173 & 404 & 926 & 1944 & 3560 & 6619 & 10274 & 12066 & 3710 & & & & 39723 \\[-5pt]
    5 & & & & 2 & 7 & 18 & 19 & 178 & 58 & 33 & 21 & & & & & & 336 \\
    \hline
    total & 2 & 13 & 46 & 156 & 671 & 1309 & 2917 & 6238 & 11459 & 20745 & 34425 & 44511 & 32301 & 9237 & 99 & 0 & 164129 \\
    \hline
  \end{tabular}
\end{center}
\begin{center}
  \renewcommand{\arraystretch}{1.4}
  \begin{tabular}{|r|rrrrrrrrrrrrrrrr|}
    \hline
    alph.\@ & sync.\@ & sync.\@ & sync.\@ & sync.\@ & sync.\@ & sync.\@ & sync.\@ & sync.\@ & sync.\@ & sync.\@ & sync.\@ & sync.\@ & sync.\@ & sync.\@ & sync.\@ & sync.\@ \\[-5pt]
    size & 25 & 24 & 23 & 22 & 21 & 20 & 19 & 18 & 17 & 16 & 15 & 14 & 13 & 12 & 11 & 10 \\
    \hline
    1 & & & & & & & & & & & & & & & & \\[-5pt]
    2 & 2 & & 2 & 11 & 22 & 45 & 61 & 112 & 201 & 322 & 528 & 954 & 1761 & 2540 & 4077 & 6341 \\[-5pt]
    3 & & & 2 & 35 & 126 & 285 & 568 & 1355 & 4801 & 12092 & 20636 & 44871 & 92738 & 174948 & 312377 & 584993 \\[-5pt]
    4 & & & & 7 & 57 & 153 & 347 & 1319 & 5789 & 16414 & 38463 & 98340 & 209987 & 411502 & 855834 & 1658196 \\[-5pt]
    5 & & & & 1 & 4 & 10 & 41 & 285 & 1035 & 2895 & 11428 & 41010 & 96178 & 179536 & 827097 & 1169501 \\[-5pt]
    6 & & & & & & & 2 & 11 & 26 & 42 & 1052 & 2925 & 1128 & 215 & 298427 & 33953 \\
    \hline
    total & 2 & 0 & 4 & 54 & 209 & 493 & 1019 & 3082 & 11852 & 31765 & 72107 & 188100 & 401792 & 768741 & 2297812 & 3452984 \\
    \hline
  \end{tabular}
\end{center}
\caption{The number of (slowly) synchronizing transitive minimal DFAs with $3$ to $6$ states, 
up to reordering states.} \label{om}
\end{minipage}}
\end{figure}

A conjecture of \^{A}ngela Cardoso asserts that the maximum subset synchronization 
lengths of the Cerny automata are the best possible for \emph{synchronizing} 
DFAs, see \cite{KKS16}. The maximum synchronization length of the Cerny automaton
with $n$ states are
$$
(n-1)^2 \Big(\Big\lceil \frac{n}{|S|} \Big\rceil - 1 \Big) 
\Big( 2n - |S| \Big\lceil \frac{n}{|S|} \Big\rceil - 1 \Big) 
$$
for a subset $S$. For nonsynchronizing DFAs, subset synchronization 
lengths can be exponential in the number of states. 

\begin{figure}[t]
\begin{center}
  \renewcommand{\arraystretch}{1.4}
  \begin{tabular}{|r|r!{\color{lightgray}\vrule}rr!{\color{lightgray}\vrule}rrr|}
    \hline
     & $n = 3$ & $n = 4$ & $n = 4$ & $n = 5$ & $n = 5$ & $n = 5$ \\[-5pt]
     & $|S| = 2$ & $|S| = 2$ & $|S| = 3$ & $|S| = 2$ & $|S| = 3$ & $|S| = 4$ \\
    \arrayrulecolor{lightgray}\cline{1-1}\cline{2-2}\cline{3-4}\cline{5-7}\arrayrulecolor{black}      
    alph.\@ & sync.\@ & sync.\@ & sync.\@ & sync.\@ & sync.\@ & sync.\@ \\[-5pt]
    size & 3 & 6 & 8 & 10 & 13 & 15 \\
    \hline
    1 & & & & & & \\[-5pt]
    2 & 6 (5) & 3 (3) & 3 (3) & 1 (1) & 2 (2) & 1 (1) \\[-5pt]
    3 & 23 (2) & 10 (4) & 11 (4) & 1 (1) & 2 (2) & 2 (2) \\[-5pt]
    4 & 30 (0) & 9 (0) & 13 (0) & & 1 (0) & 1 (0) \\[-5pt]
    5 & 20 (0) & 5 (0) & 6 (0) & & & \\[-5pt]
    6 & 7 (0) & 1 (0) & 1 (0) & & & \\[-5pt]
    7 & 1 (0) & & & & & \\
    \hline
    total & 87 (7) & 28 (7) & 34 (7) & 2 (2) & 5 (4) & 4 (3) \\
    \hline
  \end{tabular}
\end{center}
\caption{Number of DFAs with the largest subset synchronization lengths,
up to reordering states.} \label{os}
\end{figure}

We verified Cardoso's conjecture for DFAs up to $7$ states. 
In figure \ref{os}, the number of (transitive minimal) basic DFAs with $n$ states 
in which it takes the maximum number of steps to synchronize a subset of
size $|S|$, up to reordering states, is given for 
$2 \le |S| < n \le 5$.  

For $6$ states, the only basic DFAs which require the maximum number of steps
to synchronize subsets are the Cerny automaton with $6$ states and the 
Kari automaton, the latter of which for $|S| \ge 4$ only. For $7$ states, 
the only DFA which requires the maximum number of steps to synchronize 
subsets is the Cerny automaton. So it seems plausible
that for $n \ge 7$ states, the Cerny automaton is the only automaton 
which reaches the Cardoso bound.

\section{Maximal and semi-minimal synchronizing DFAs} \label{max}

In \cite{BDZ16}, we counted the number of basic synchronizing DFAs for up to $6$ states 
and large synchronization lengths. We reduced the synchronization lengths until the number
of basic synchronizing DFAs became too large.

To deal better with finding many synchronizing DFAs, we made two improvements to 
the search algorithm. In the search algorithm, the candidate symbols for extension 
are sorted in order of increasing number of synchronizing pairs. But this does not 
do anything if the DFA of the symbols that we have already chosen is synchronizing.
The first improvement is to sort the symbols as well if the DFA of the symbols that we 
have already chosen is synchronizing. The candidate symbols for extension are sorted
in order of increasing synchronization length.

The second improvement deals with the symmetry reduction of the synchronizing DFAs
which are found be the algorithm. The algorithm itself performs symmetry reduction
as an optimization, but this symmetry reduction is not perfect. But we need perfect 
symmetry reduction for for finding canonical representations to be stored and 
counting. This is done by applying all $n!$ symmetries on all symbols on candidate 
new synchronizing DFAs, where $n$ is the number of states. But applying symmetries
on symbols takes some time. A lookup table for the symmetry applications would
require $n! \cdot n^n$ entries for $n$ states, which makes it too large. 
For that reason, we ordered the symmetries with the Johnson-Trotter algorithm 
for each $n$, reducing the size of the lookup table to only $(n-1) \cdot n^n$ 
entries for $n$ states.

But these improvements do not solve the problem that there are too many synchronizing 
DFAs. To deal with that problem, we only searched for DFAs with additional properties
for smaller synchronization bounds. A synchronizing basic DFA is
\begin{itemize}

\item \emph{minimal}, if it becomes nonsynchronizing after removing any symbol. 

\item \emph{semi-minimal}, if its synchronization length increases or it becomes 
nonsynchronizing after removing any symbol. 

\item \emph{maximal}, if its synchronization length decreases after adding any 
new symbol.

\end{itemize}
Here, a symbol is new if it acts differently on the set of states.

For these types of DFAs, the number of DFAs appeared not to be very large even for
smaller synchronization lengths. We counted the different types of synchronizing
DFAs by testing found DFAs on having the type. With this, we kept track of symbols
for the test for maximality, because testing all symbols takes very long.
But we also optimized the search process. With the minimal DFAs, 
we did not search through synchronizing DFAs,
because extensions of synchronizing DFAs are not minimal. 

With the semi-minimal, maximal, and combined types, the collection of found DFAs
is moved to another place in the code, namely to the new procedure described
above, which sorts the symbols if the DFA is already synchronizing. 
The synchronizing DFA itself is collected as a candidate for a semi-minimal
DFA. The candidate maximal DFA is made by saturating the synchronizing DFA 
with the sorted symbols, in such a way that the synchronization length is not 
affected. Next, the search process is continued, but extensions within the
saturated DFA are skipped.

In the tables below, we do not give the number of DFAs, but we gives ranges of 
possible alphabet sizes, for minimal, semi-minimal, maximal, maximal minimal, and
maximal semi-minimal DFAs with a specific state set and synchronization length.
The number of DFAs for each alphabet size in such a range can be found with the 
source code. Ranges are given for synchronizing DFAs which do not need to be 
transitive, but we verified that the corresponding ranges for transitive DFAs 
can be obtained by removing $1$ (if present).

The ranges for general synchronizing basic DFAs were found as follows. Suppose that
$\mathcal{B}$ is a \emph{maximum} DFA, i.e.\@ a maximal DFA with the largest
possible alphabet size. By removing symbols of $\mathcal{B}$, we can obtain a 
semi-minimal DFA $\mathcal{A}$ with the same synchronization length as $\mathcal{B}$.
Consequently, to conclude that the range for general synchronizing basic DFAs
with $\mathcal{A}$ and $\mathcal{B}$ is continuous, it suffices to verify that the
the range of semi-minimal synchronizing DFAs with $\mathcal{A}$ is continuous.

But this does not work for the ranges of transitive general synchronizing basic DFAs,
because $\mathcal{A}$ may be not transitive. However, for the actual maximum DFAs 
$\mathcal{B}$ which were printed by the search algorithm, it appeared that
it was possible to make $\mathcal{A}$ transitive by restoring one symbol of
$\mathcal{B}$. So the ranges for general basic DFAs can be deduced from the maximal 
and semi-minimal ranges, and the corresponding transitive ranges can be obtained by 
removing $1$ (if present), just as for the other ranges.

In figure \ref{234}, we give the results for up to $4$ states. We were able to
get through down to synchronization length $1$.

\begin{figure}
\begin{center}
  \renewcommand{\arraystretch}{1.4}
  \begin{tabular}{|c|ccc!{\color{lightgray}\vrule}ccc|}
    \hline
    sync.\@ & & & & max & max & max \\[-5pt]
    length & \hphantom{min} & min & smin & \hphantom{min} & min & smin \\
    \hline
    1 & 1--3 & 1 & 1 & 3 &  &  \\
    \hline
    all & 1--3 & 1 & 1 & 3 &  &  \\
    \hline
  \end{tabular}
\end{center}
\begin{center}
  \renewcommand{\arraystretch}{1.4}
  \begin{tabular}{|c|ccc!{\color{lightgray}\vrule}ccc|}
    \hline
    sync.\@ & & & & max & max & max \\[-5pt]
    length & \hphantom{min} & min & smin & \hphantom{min} & min & smin \\
    \hline
    4 & 2--5 & 2--3 & 2--3 & 5 &  &  \\[-5pt] 
    3 & 2--9 & 2 & 2 & 9 &  &  \\[-5pt]
    2 & 1--23 & 1--2 & 1--2 & 23 &  &  \\[-5pt]
    1 & 1--26 & 1 & 1 & 26 &  &  \\
    \hline
    all & 1--26 & 1--3 & 1--3 & 5,~~9,~~23,~~26 &  &  \\
    \hline
  \end{tabular}
\end{center}
\begin{center}
  \renewcommand{\arraystretch}{1.4}
  \begin{tabular}{|c|ccc!{\color{lightgray}\vrule}ccc|}
    \hline
    sync.\@ & & & & max & max & max \\[-5pt]
    length & \hphantom{min} & min & smin & \hphantom{min} & min & smin \\
    \hline
    9 & 2--5 & 2--3 & 2--3 & 2--3,~~5 & 2--3 & 2--3 \\[-5pt] 
    8 & 2--8 & 2--4 & 2--4 & 4--8 &  &  \\[-5pt]
    7 & 2--17 & 2--4 & 2--4 & 2,~~4--9,~~17 & 2 & 2 \\[-5pt]
    6 & 2--17 & 2--4 & 2--4 & 5--11,~~13--15,~~17 &  &  \\[-5pt]
    5 & 2--41 & 2--4 & 2--5 & 11--25,~~35,~~41 &  &  \\[-5pt]
    4 & 2--59 & 2--3 & 2--4 & 
        23, 25, 27, \dots, 51, 53,~~59
    &  &  \\[-5pt]
    3 & 1--167 & 1--3 & 1--3 & \begin{tabular}{@{}c@{}}
        79,~~83,~~91,~~101,~~103, \\[-8pt] 119,~~123,~~127,~~147,~~167
    \end{tabular} &  &  \\[-5pt]
    2 & 1--251 & 1--2 & 1--2 & 251 &  &  \\[-5pt]
    1 & 1--255 & 1 & 1 & 255 &  &  \\
    \hline
    all & 1--255 & 1--4 & 1--5 & \begin{tabular}{@{}c@{}}
        2--25,~~27, 29, 31, \dots, 51, 53, \\[-8pt]
        59,~~79,~~83,~~91,~~101,~~103, \\[-8pt]
        119,~~123,~~127,~~147,~~167, \\[-8pt]
        251,~~255 
    \end{tabular} & 2--3 & 2--3 \\
    \hline
  \end{tabular}
\end{center}
\caption{Alphabet size ranges for synchronizing basic DFAs with $2$, $3$ and $4$ states.} \label{234}
\end{figure}

\begin{figure}[p] 
\begin{center}
  \renewcommand{\arraystretch}{1.4}
  \begin{tabular}{|c|ccc!{\color{lightgray}\vrule}ccc|}
    \hline
    sync.\@ & & & & max & max & max \\[-5pt]
    length & \hphantom{min} & min & smin & \hphantom{min} & min & smin \\
    \hline
    16 & 2--3 & 2--3 & 2--3 & 2--3 & 2--3 & 2--3 \\[-5pt]
    15 & 2--6 & 2--4 & 2--4 & 2--6 & 2--3 & 2--3 \\[-5pt]
    14 & 2--13 & 2--4 & 2--4 & 2--8,~~13 & 2--3 & 2--3 \\[-5pt]
    13 & 2--15 & 2--5 & 2--5 & 2--10,~~12--13,~~15 & 2--4 & 2--4 \\[-5pt]
    12 & 2--23 & 2--5 & 2--5 & 2--17,~~19--21,~~23 & 2--3 & 2--3 \\[-5pt]
    11 & 2--29 & 2--5 & 2--6 & 2--25,~~27,~~29 & 2--4 & 2--4 \\[-5pt]
    10 & 2--71 & 2--5 & 2--6 & 2--27,~~29,~~31,~~71 & 2--3 & 2--3 \\[-5pt]
    9 & 2--71 & 2--5 & 2--7 & \begin{tabular}{@{}c@{}}
        2--41,~~43--47,~~49--51, \\[-8pt] 53,~~55,~~57,~~59,~~71 
    \end{tabular} & 2--3 & 2--4 \\[-5pt]
    8 & 2--89 & 2--5 & 2--7 & \begin{tabular}{@{}c@{}}
        3--57,~~59--71,~~73--75, \\[-8pt] 77,~~83,~~89 
    \end{tabular} &  & 3 \\[-5pt]
    7 & 2--215 & 2--5 & 2--7 & \begin{tabular}{@{}c@{}}
        4--85,~~87--89,~~91--99, \\[-8pt] 101,~~105,~~167,~~215 
    \end{tabular} &  &  \\[-5pt]
    6 & ? & 2--5 & 2--6 & ? &  & ? \\[-5pt]
    5 & ? & 2--4 & 2--5 & ? &  & ? \\[-5pt]
    4 & ? & 1--4 & 1--4 & ? &  & ? \\[-5pt]
    3 & ? & 1--3 & 1--3 & ? &  & ? \\[-5pt]
    2 & 1--3119 & 1--2 & 1--2 & 3119 &  &  \\[-5pt]
    1 & 1--3124 & 1 & 1 & 3124 &  &  \\
    \hline
    all & 1--3124 & 1--5 & 1--7 & ? & 2--4 & ? \\
    \hline
  \end{tabular}
\end{center}
\begin{center}
  \renewcommand{\arraystretch}{1.4}
  \begin{tabular}{|c|ccc!{\color{lightgray}\vrule}ccc|}
    \hline
    sync.\@ & & & & max & max & max \\[-5pt]
    length & \hphantom{min} & min & smin & \hphantom{min} & min & smin \\
    \hline
    25 & 2 & 2 & 2 & 2 & 2 & 2 \\[-5pt]
    24 &  &  &  &  &  &  \\[-5pt]
    23 & 2--3 & 2--3 & 2--3 & 2--3 & 2--3 & 2--3 \\[-5pt]
    22 & 2--11 & 2--5 & 2--5 & 2--7,~~10--11 & 2--3 & 2--4 \\[-5pt]
    21 & 2--15 & 2--5 & 2--5 & 2--15 & 2--4 & 2--4 \\[-5pt]
    20 & 2--21 & 2--5 & 2--6 & 2--17,~~19,~~21 & 2--4 & 2--4 \\[-5pt]
    19 & 2--47 & 2--6 & 2--6 & 2--17,~~19,~~25,~~27,~~47 & 2--3 & 2--5 \\[-5pt]
    18 & 2--53 & 2--6 & 2--6 & 2--25,~~47,~~53 & 2--4 & 2--5 \\[-5pt]
    17 & 2--59 & 2--6 & 2--7 & \begin{tabular}{@{}c@{}}
        2--29,~~31--33,~~35,~~37, \\[-8pt] 39,~~41,~~43,~~45,~~59 
    \end{tabular} & 2--4 & 2--5 \\[-5pt]
    16 & 2--95 & 2--6 & 2--7 & \begin{tabular}{@{}c@{}}
        2--41,~~43,~~45,~~47,~~49, \\[-8pt] 51,~~53,~~59,~~61,~~65, \\[-8pt] 77,~~79,~~83,~~89,~~95 
    \end{tabular} & 2--4 & 2--6 \\[-5pt]
    15 & 2--101 & 2--6 & 2--8 & 2--71,~~75,~~77,~~80--85,~~101 & 2--4 & 2--5 \\[-5pt]
    14 & 2--143 & 2--6 & 2--9 & \begin{tabular}{@{}c@{}}
        2--93,~~95--105,~~107,\\[-8pt] 113,~~115,~~119,~~123,~~125,\\[-8pt] 127,~~131,~~137,~~143 
    \end{tabular} & 2--5 & 2--5 \\[-5pt]
    13 & ? & 2--6 & ? & ? & 2--4 & ? \\[-5pt]
    12 & ? & 2--6 & ? & ? & ? & ? \\[-5pt]
    11 & ? & 2--6 & ? & ? & ? & ? \\[-5pt]
    10 & ? & 2--6 & ? & ? & ? & ? \\
    \hline
  \end{tabular}
\end{center}
\caption{Alphabet size ranges for synchronizing basic DFAs with $5$ and $6$ states.} \label{56}
\end{figure}

For $5$ states, we were able to get through only for minimal DFAs. For $6$ states,
we were not able to get through at all. The results are given in figure \ref{56}.
Notice that some additional ranges are given in the table for $5$ states as well.
The lines for synchronization lengths $1$ and $2$ were obtained by reasoning. 
This reasoning can be generalized to any number of states. The
maximal minimal ranges were obtained by testing minimal DFAs for
maximality, which was done by an algorithm to test the procedure of
keeping track of the symbols for the test for maximality 
(not included in the source code).

Finally, we describe how we found the ranges for semi-minimal synchronizing
DFAs with $5$ states. Notice first that these ranges contain the corresponding 
ranges for minimal DFAs, that non-minimal semi-minimal synchronizing DFAs have
at least $2$ symbols, and that the number of symbols of a semi-minimal 
synchronizing DFAs does not exceed its synchronization length. This yields
the validity of the ranges for synchronization length $\le 4$. Although the 
algorithm did not complete synchronization length $6$, it did find 
semi-minimal synchronizing DFA with synchronization length $6$ and up to
$6$ symbols. This yields the validity of the range for synchronization length $6$.
To complete the range for synchronization length $5$, we need a construction 
with $5$ symbols, which is given below, where self-transitions are omitted.
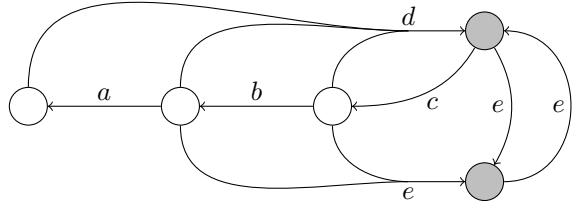
\begin{center}
\begin{tikzpicture}
\useasboundingbox (-0.25,-1.25) rectangle (7.15,1.38);
\tikzstyle{nodestyle}=[draw,fill=white,circle,inner sep=0pt,minimum width=0.5cm]
\node[nodestyle] (0) at (0,0) {};
\node[nodestyle] (1) at (2,0) {};
\node[nodestyle] (2) at (4,0) {};
\node[nodestyle,fill=lightgray] (3) at (6,1) {};
\node[nodestyle,fill=lightgray] (4) at (6,-1) {};
\coordinate (3d) at (5,1);
\node[anchor=south,inner sep=2pt] at (3d) {$d$};
\draw[out=90,in=180] (0) edge (3d) (1) edge (3d) (2) edge (3d);
\coordinate (4e) at (5,-1);
\node[anchor=north,inner sep=2pt] at (4e) {$e$};
\draw[out=-90,in=-180] (1) edge (4e) (2) edge (4e);
\draw (3d) edge[->] (3) (4e) edge[->] (4);
\draw[out=-120,in=0] (3) edge[->] node[auto,inner sep=1pt] {$c$} (2);
\draw[out=-60,in=60] (3) edge[->] node[left,inner sep=2.5pt] {$e$} (4);
\draw[out=0,in=0,looseness=1.5] (4) edge[->] node[left,inner sep=2pt] {$e$} (3);
\draw (2) edge[->] node[above,inner sep=2pt] {$b$} (1) 
      (1) edge[->] node[above,inner sep=2pt] {$a$} (0);  
\end{tikzpicture}
\end{center}
The shaded pair of states requires $5$ steps to synchronize, and the other 
states synchronize as well. The construction can be generalized to $n \ge 4$ 
states, with $n$ steps and $n$ symbols (the construction is not 
semi-minimal for $3$ states).

Synchronization length $13$ is only included in the table for minimal synchronizing 
DFAs with $6$ states. But we think the maximum DFA with $6$ states and synchronization length 
$13$ has $359$ symbols. More generally, we think the maximum DFA with $n$ states and 
synchronization length $3n-5$ has $3 \cdot (n-1)! - 1$ symbols. 

This number of symbols is indeed obtainable. Take a state set $Q$ of size $n$,
with distinct states $q$ and $q'$. We include (i) all $(n-1)!$ symbols which send $Q$ 
to $Q$ and $q$ to $q$, except the identity symbol, (ii) all $(n-1)!$ symbols which send 
$Q$ to $Q$ and $q$ to $q'$, and (iii) all $(n-1)!$ symbols which send $Q$ to 
$Q \setminus \{q\}$ and which send $q$ and $q'$ to the same state.

\begin{figure}[b]
\begin{center}
  \renewcommand{\arraystretch}{1.4}
  \begin{tabular}{|c|ccc!{\color{lightgray}\vrule}ccc|}
    \hline
    sync.\@ & & & & max & max & max \\[-5pt]
    length & \hphantom{min} & min & smin & \hphantom{min} & min & smin \\
    \hline
    36 & 2 & 2 & 2 & 2 & 2 & 2 \\[-5pt]
    35 &  &  &  &  &  &  \\[-5pt]
    34 &  &  &  &  &  &  \\[-5pt]
    33 &  &  &  &  &  &  \\[-5pt]
    32 & 2--3 & 2 & 2 & 3 &  &  \\[-5pt]
    31 & 2--4 & 2--3 & 2--3 & 2--4 & 2--3 & 2--3 \\[-5pt]
    30 & 2--11 & 2--6 & 2--6 & 2--6,~~8--11 & 2--3 & 2--4 \\[-5pt]
    29 & 2--21 & 2--6 & 2--6 & 2--15,~~17,~~21 & 2--4 & 2--4 \\
    \hline
  \end{tabular}
\end{center}
\end{figure}

For $7$ states, the idea was to start a search process to find all maximal
and semi-minimal DFAs with synchronization length at least $27$. A sample
of $5$ percent of this computation on a heterogeneous cluster indicated that
this takes about $45$ CPU years on that cluster (of which $2$ years are already
completed by the sample). But we did not get the time to do the whole computation.
For that reason, I wrote a program to extract the maximal and semi-minimal DFAs
with synchronization length at least $29$ from all basic DFAs with synchronization 
length at least $29$. The selection of the maximal DFAs requires two passes. 
In the first pass, non-maximal DFAs are collected, by testing DFAs with one symbol 
removed to have the same synchronization length, for each DFA and each of its symbols.

Below are the alphabet size ranges for subset synchronization. 
\begin{center}
  \renewcommand{\arraystretch}{1.4}
  \begin{tabular}{|c!{\color{lightgray}\vrule}c!{\color{lightgray}\vrule}c|ccc!{\color{lightgray}\vrule}ccc|}
    \hline
    & & sync.\@ & & & & max & max & max \\[-5pt]
    $n$ & $|S|$ & length & \hphantom{min} & min & smin & \hphantom{min} & min & smin \\
    \hline
    2 & 2 & 1 & 1--3 & 1 & 1 & 3 &  &  \\
    \arrayrulecolor{lightgray}\cline{1-3}\cline{4-9}\arrayrulecolor{black} 
    3 & 2 & 3 & 2--7 & 2--3 & 2--3 & 7 &  &  \\[-5pt]
    3 & 3 & 4 & 2--5 & 2--3 & 2--3 & 5 &  &  \\
    \arrayrulecolor{lightgray}\cline{1-3}\cline{4-9}\arrayrulecolor{black} 
    4 & 2 & 6 & 2--6 & 2--3 & 2--3 & 3,~~6 & 3 & 3 \\[-5pt]
    4 & 3 & 8 & 2--6 & 2--3 & 2--3 & 2--3,~~5--6 & 2--3 & 2--3 \\[-5pt]
    4 & 4 & 9 & 2--5 & 2--3 & 2--3 & 2--3,~~5 & 2--3 & 2--3 \\
    \arrayrulecolor{lightgray}\cline{1-3}\cline{4-9}\arrayrulecolor{black} 
    5 & 2 & 10 & 2--3 & 2--3 & 2--3 & 2--3 & 2--3 & 2--3 \vspace*{-5pt} \\
    5 & 3 & 13 & 2--4 & 2--3 & 2--3 & 2,~~4 & 2 & 2 \\[-5pt]
    5 & 4 & 15 & 2--4 & 2--3 & 2--3 & 2,~~4 & 2 & 2 \\[-5pt]
    5 & 5 & 16 & 2--3 & 2--3 & 2--3 & 2--3 & 2--3 & 2--3 \\
    \arrayrulecolor{lightgray}\cline{1-3}\cline{4-9}\arrayrulecolor{black} 
    6 & 2 & 15 & 2 & 2 & 2 & 2 & 2 & 2 \\[-5pt]
    6 & 3 & 20 & 2 & 2 & 2 & 2 & 2 & 2 \\[-5pt]
    6 & 4 & 22 & 2 & 2 & 2 & 2 & 2 & 2 \\[-5pt]
    6 & 5 & 24 & 2 & 2 & 2 & 2 & 2 & 2 \\[-5pt]
    6 & 6 & 25 & 2 & 2 & 2 & 2 & 2 & 2 \\
    \arrayrulecolor{lightgray}\cline{1-3}\cline{4-9}\arrayrulecolor{black} 
    7 & 2 & 21 & 2 & 2 & 2 & 2 & 2 & 2 \\[-5pt]
    7 & 3 & 28 & 2 & 2 & 2 & 2 & 2 & 2 \\[-5pt]
    7 & 4 & 31 & 2 & 2 & 2 & 2 & 2 & 2 \\[-5pt]
    7 & 5 & 33 & 2 & 2 & 2 & 2 & 2 & 2 \\[-5pt]
    7 & 6 & 35 & 2 & 2 & 2 & 2 & 2 & 2 \\[-5pt]
    7 & 7 & 36 & 2 & 2 & 2 & 2 & 2 & 2 \\
    \hline
  \end{tabular}
\end{center}

We can observe the following in the results.

\begin{conjecture}
Let $n \ge 3$.
\begin{enumerate}

\item[\upshape(i)] The maximum number of symbols of a minimal DFA with $n$ states is $n$.
This number of symbols is possible for minimal DFAs with $n$ states, if and only
if the synchronization length is at least $n+1$ and at most 
$\frac12 n^2 + \frac12 n - 2$.

\item[\upshape(ii)] The maximum number of symbols of a semi-minimal DFA with $n$ states is $2n-3$.
If $n \ge 4$, then this number of symbols is possible for semi-minimal DFA with $n$ 
states, if and only if the synchronization length is at least $2n-3$ and at most 
$\frac12 n^2 - \frac12 n - 1$.

\end{enumerate}
Furthermore, transitive constructions are possible.
\end{conjecture}

We show that transitive minimal DFAs with $n$ states and $n$ symbols as in (i) above indeed 
exist. Below on the left hand side, a construction is given for synchronization length
$\frac12 n^2 + \frac12 n - 2$. Here, a single arrow represents a symbol which merges
two states as indicated by the arrow, and preserves the other states. Furthermore,
a double arrow represents a symbol which interchanges two states and preserves the 
other states.
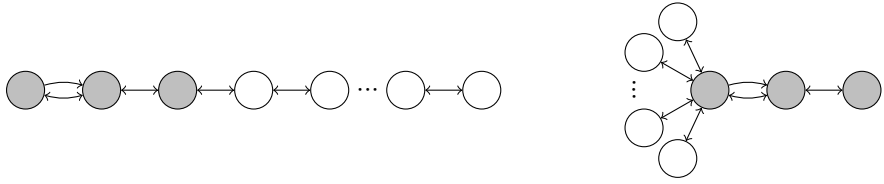
\begin{center}
\begin{tikzpicture}
\tikzstyle{nodestyle}=[draw,fill=white,circle,inner sep=0pt,minimum width=0.5cm]
\node[nodestyle,fill=lightgray] (0) at (-9,0) {};
\node[nodestyle,fill=lightgray] (1) at (-8,0) {};
\node[nodestyle,fill=lightgray] (2) at (-7,0) {};
\node[nodestyle] (3) at (-6,0) {};
\node[nodestyle] (4) at (-5,0) {};
\node[nodestyle] (5) at (-4,0) {};
\node[nodestyle] (6) at (-3,0) {};
\draw[<->] (0) edge[out=-15,in=-165] (1) (0) edge[->,out=15,in=165] (1) 
(1) edge (2) (2) edge (3) (3) edge (4) (5) edge (6);
\draw (-4.6,0) node {$\cdot$} (-4.5,0) node {$\cdot$} (-4.4,0) node {$\cdot$};
\node[nodestyle,fill=lightgray] (10) at (0,0) {};
\node[nodestyle,fill=lightgray] (11) at (1,0) {};
\node[nodestyle,fill=lightgray] (12) at (2,0) {};
\node[nodestyle] (13) at (115:1) {};
\node[nodestyle] (14) at (150:1) {};
\node[nodestyle] (15) at (-150:1) {};
\node[nodestyle] (16) at (-115:1) {};
\draw[<->] (10) edge[out=-15,in=-165] (11) (10) edge[->,out=15,in=165] (11) 
(11) edge (12) (13) edge (10) (14) edge (10) (15) edge (10) (16) edge (10);
\draw (175:1) node {$\cdot$} (180:1) node {$\cdot$} (185:1) node {$\cdot$};
\end{tikzpicture}
\end{center}
The white states can be moved to the left step by step, where each step
yields a DFA of which the synchronization length is one less than that of its 
predecessor. This process end with the DFA above on the right hand side,
which has synchronization length $2n - 2$. 
By replacing the double arrows which attach the white states by single arrows 
towards the leftmost shaded state, one can decrease the synchronization 
length further and obtain all remaining synchronization lengths down to 
$n+1$ inclusive. 

But that construction is not transitive. For a transitive construction, we
start with the semi-minimal DFA which we constructed before. This DFA is not
minimal, because symbol $d$ is not needed for synchronization: removing 
symbol $d$ yields a DFA with a sink state which synchronizes in $2n - 3$
steps. Below on the left hand side, we attached a new state with a double 
arrow to the sink state, which we marked with an $*$. We will show that 
this new DFA is minimal with synchronization length $n + 1$. 
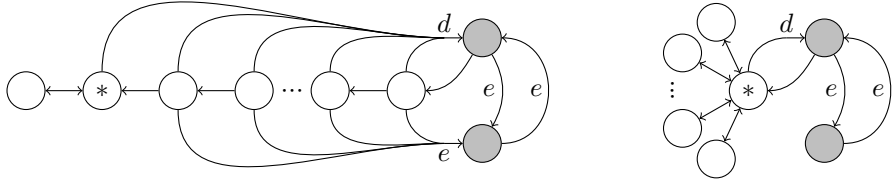
\begin{center}
\begin{tikzpicture}
\tikzstyle{nodestyle}=[draw,fill=white,circle,inner sep=0pt,minimum width=0.5cm]
\useasboundingbox (-9.75,-1.16) rectangle (1.87,1.175);
\node[nodestyle] (0a) at (-9.5,0) {};
\node[nodestyle] (0) at (-8.5,0) {$*$};
\node[nodestyle] (1a) at (-7.5,0) {};
\node[nodestyle] (1) at (-6.5,0) {};
\node[nodestyle] (2a) at (-5.5,0) {};
\node[nodestyle] (2) at (-4.5,0) {};
\node[nodestyle,fill=lightgray] (3) at (-3.5,0.7) {};
\node[nodestyle,fill=lightgray] (4) at (-3.5,-0.7) {};
\draw (-6.1,0) node {$\cdot$} (-6.0,0) node {$\cdot$} (-5.9,0) node {$\cdot$};
\coordinate (3d) at (-4,0.7);
\node[anchor=south,inner sep=2pt] at (3d) {$d$};
\draw[overlay,out=90,in=180] (0) edge (3d) (1a) edge (3d) (1) edge (3d) (2a) edge (3d) (2) edge (3d);
\coordinate (4e) at (-4,-0.7);
\node[anchor=north,inner sep=2pt] at (4e) {$e$};
\draw[overlay,out=-90,in=-180] (1a) edge (4e) (1) edge (4e) (2a) edge (4e) (2) edge (4e);
\draw (3d) edge[->] (3) (4e) edge[->] (4) (0a) edge[<->] (0);
\draw[out=-120,in=0] (3) edge[->] (2);
\draw[out=-60,in=60] (3) edge[->] node[left,inner sep=2.5pt] {$e$} (4);
\draw[out=0,in=0,looseness=1.5] (4) edge[->] node[left,inner sep=2pt] {$e$} (3);
\draw[->] (2) edge (2a) (1) edge (1a) (1a) edge (0);
\node[nodestyle] (10) at (0,0) {$*$};
\node[nodestyle,fill=lightgray] (11) at (1,0.7) {};
\node[nodestyle,fill=lightgray] (12) at (1,-0.7) {};
\node[nodestyle] (13) at (115:1) {};
\node[nodestyle] (14) at (150:1) {};
\node[nodestyle] (15) at (-150:1) {};
\node[nodestyle] (16) at (-115:1) {};
\draw[<->] (13) edge (10) (14) edge (10) (15) edge (10) (16) edge (10);
\draw (175:1) node {$\cdot$} (180:1) node {$\cdot$} (185:1) node {$\cdot$};
\coordinate (11d) at (0.5,0.7);
\node[anchor=south,inner sep=2pt] at (11d) {$d$};
\draw[out=90,in=180] (10) edge (11d);
\draw (11d) edge[->] (11);
\draw[out=-120,in=0] (11) edge[->] (10);
\draw[out=-60,in=60] (11) edge[->] node[left,inner sep=2.5pt] {$e$} (12);
\draw[out=0,in=0,looseness=1.5] (12) edge[->] node[left,inner sep=2pt] {$e$} (11);
\end{tikzpicture}
\end{center}
One can attach more new states on state $*$ with double arrows, up to the DFA 
above on the right hand side. We will show that we obtain all synchronization
lengths from $n + 2$ up to $2n - 3$ inclusive this way.

Just as before, the objective is to merge the shaded pair of states. But there
is a second objective, namely to apply the interchange symbols. To make the first
application of the interchange symbols effective, they have to be preceded by another
interchange symbol or by symbol $d$, and we may assume the latter symbol to be the 
direct predecessor of the former. But a consecutive application of two interchange
symbols will not occur in a shortest synchronizing word. So the second objective is
that for each of the interchange symbols, there is an application which is
immediately after symbol $d$.

Let $k$ be the number of states which is attached to state $*$ with an interchange
symbol. To show that the length of the shortest synchronizing word is $n + k$,
we need a third objective, which is that the last symbol is not an interchange
symbol. This objective is justified because interchange symbols act as permutations
on the state set, and therefore cannot be the last symbol of a shortest synchronizing 
word. Each of the time, an application of symbol $d$ does not 
contribute to the merge of the shaded pair of states, and neither do
interchange symbols, except in the last step where the actual merge takes place 
by way of symbol $d$. This exception is compensated by the third
objective. It is also clear that a synchronizing word of length $n + k$ exists,
so we have all synchronization lengths from $n + 1$ up to $2n - 3$ inclusive.

\end{document}